\documentclass[a4paper,12pt]{article}

\usepackage{color,amssymb,mathrsfs,amsmath}

\usepackage{cite}

\newcommand{\ket}[1]{|#1\rangle}
\newcommand{\bra}[1]{\langle#1|}

\DeclareRobustCommand\openone{\leavevmode\hbox{\small1\normalsize\kern-.33em1}}

\newcommand{\tr}{\mathrm{tr}}
\newtheorem{theorem}{Theorem}
\newtheorem{lemma}{Lemma}
\newenvironment{proof}{\textit{Proof.---}}{\hspace*{\fill}$\square$\\[0.4em]}

\begin{document}

\begin{center}
{\LARGE\bf On entropy growth and the hardness of simulating time
evolution}

\vspace{0.5cm}

{\large N.\ Schuch, M.M.\ Wolf,
K.G.H.\ Vollbrecht, and
J.I.\ Cirac}

\vspace{0.3cm}

\small Max-Planck-Institut f\"ur Quantenoptik,\\ Hans-Kopfermann-Str.\ 1,
D-85748 Garching, Germany.

\vspace{1cm}
{\large Abstract}
\vspace{0.3cm}

\parbox{11.5cm}{
The simulation of quantum systems is a task for which quantum
computers are believed to give an exponential speedup as compared
to classical ones. While ground states of one-dimensional systems
can be efficiently approximated using Matrix Product States (MPS),
their time evolution can encode quantum computations, so that
simulating the latter should be hard classically. However, one
might believe that for systems with high enough symmetry, and thus
insufficient parameters to encode a quantum computation, efficient
classical simulation is possible. We discuss supporting evidence
to the contrary: We provide a rigorous proof of the observation
that a time independent local Hamiltonian can yield a linear increase of
the entropy when acting on a product state in a translational invariant
framework.  This criterion has to be met by any
classical simulation method, which
in particular
implies that every global approximation of
the evolution requires exponential resources for any MPS based method.}
\end{center}

\vspace{.5cm}

\section{Introduction}

One of the main challenges in the field of Quantum Computation is
to determine the boundaries between classical and quantum
computers and pinpoint the computational problems for which
quantum devices will be more powerful than classical ones. The
research on this topic goes in two directions: On the one hand,
several efficient quantum algorithms for mathematical problems
which are believed to be hard classically have been found, the
most prominent being Shor's celebrated algorithm for factoring and
discrete logarithms~\cite{shoralg}.
 On the other hand, the idea
to use quantum simulators to simulate quantum systems has gained
growing attention, as simulating quantum systems is arguably the
most natural thing to do with a quantum computer~\cite{lloyd}.
Moreover, in the near future quantum devices are far more likely
to be used for simulating quantum systems than for factoring large
numbers, both due to the much smaller number of qubits required
and the weaker demands on the control of the system. In
particular, there is no need to have a universal quantum computer
available~\cite{sorensen,jane,porras}.

In fact, while universal quantum computation requires a dynamics
which is either inhomogeneous in
time~\cite{vollbrecht1,raussendorf} or in
space~\cite{vollbrecht2,kay} (or both, as in the standard network
model), physically interesting dynamics to simulate can be
homogeneous in time, i.e., generated by a time-independent
Hamiltonian, and translational invariant.

Regarding static properties translational invariance appears to
considerably simplify the problem and enables us in many cases to
find efficient classical approximation algorithms. For instance
ground states of local one-di\-men\-sio\-nal Hamiltonians, in
particular in the presence of an energy gap, have an efficient
classical description~\cite{hastings1,hastings2} in terms of
so-called Matrix Product States (MPS)~\cite{mps-def}, which are
used as a variational ansatz in the the Density Matrix
Renormalization Group (DMRG) algorithm~\cite{white,schollwoeck}.

If static properties of a quantum system are efficiently
simulatable on a classical computer in the presence of
translational symmetry, why not also time evolution governed by a
time-independent local Hamiltonian?

The present paper is devoted to providing evidence for the
hardness of classically simulating quantum mechanical time
evolution for systems which are homogeneous in time \emph{and}
space. Hence, quantum computers (or simulators) may indeed yield
an exponential speed-up compared to their classical counterparts.

Clearly, we cannot hope for a complexity theoretic separation between
classical an quantum due to the lack of parameters needed to encode a
computational task. We will rather show first that a sufficiently fast
(e.g. linear) increase of the entanglement entropy makes it hard -- and
for MPS based approaches~\cite{vidal,osborne} impossible -- to even
approximate the state of the system efficiently. That such a linear
increase can occur even in the simplest one-dimensional systems was
observed in~\cite{cardy} based on arguments of conformal field theory
together with numerical diagonalization of chains up to 100 sites. The
main part of the present paper is devoted to a rigorous proof of this
observation.

\section{The hardness of simulating time evolution}

We show the hardness of simulating time evolution using MPS in two steps:
In a first stage,  we take a particular system initially in a
translational invariant product state of spin-$\tfrac12$ particles, and
let it evolve under a translational invariant nearest neighbor
Hamiltonian. For this system, we prove a linear lower bound to the entropy
of any contiguous block as a function of time, as observed by Calabrese
and Cardy in~\cite{cardy}.  This evolution has also been studied in
\cite{eisert}, where similar effects were observed.
We then combine this linear bound with the
continuity inequality for the von Neumann entropy which leads to a linear
lower bound on the block entropy of any approximation of the
state~\cite{mps-entropies}.  As the computational effort needed to deal
with a general MPS grows exponentially in the entropy, it would thus
increase exponentially with the time $t$ to be simulated.

In the following, logarithms are understood to the base $2$. In
particular, we define the von Neumann entropy $S(\rho)=-\tr[\rho\log\rho]$
with the base $2$ logarithm.   The natural logarithm will be denoted by
$\ln$.

\subsection{Setup and entropy scaling}

\begin{theorem}
\label{thm:entropy}
Consider a chain with an odd number $N$ of spins with periodic boundary
conditions, corresponding to a Hilbert space $(\mathbb C^2)^{\otimes N}$,
which at time $t=0$ is in the state
$\ket{\psi(0)}=\ket{\uparrow\dots\uparrow}\equiv\ket{1\dots1}$ and
evolves under the Ising Hamiltonian\footnote{
We choose $\sigma_x$, $\sigma_z$ the Pauli matrices with eigenvalues
$\pm1$.
}
\begin{equation}
\label{eq:spinham}
\mathcal H=-\frac12\sum_{j}\left[\sigma^x_j\sigma^x_{j+1}+\sigma^z_j\right]\ ,
\end{equation}
so that $\ket{\psi(t)}=e^{-i\mathcal Ht}\ket{\psi(0)}$. Define
\[S_L(t)=-\tr[\rho_L(t)\log\rho_L(t)]\]
with $\rho_L(t)=\tr_{L+1,\dots,N}\ket{\psi(t)}\bra{\psi(t)}$.
Then for $N\ge20$, $L\ge10$, $4\le t\le eL/4$, and $t\le N/5$,
\[
S_L(t)\ge \frac{4}{3\pi}t -\tfrac12\ln t -1\ .
\]
\end{theorem}
We postpone the proof of the Theorem to
Sec.~\ref{sec:proof}.

Note that although the lower bound is independent on $L$ -- the intuition
behind being that entanglement arises at the boundaries --
the maximally attainable
entanglement is proportional to the length of the block due to the
constraint $t\le eL/4$.

\subsection{Exponential scaling of bond dimension}

Let us now see what the entropy scaling implies for an ansatz state
$\ket{\phi}\equiv\ket{\phi(t)}$
which we use to approximate
$\ket{\psi}\equiv\ket{\psi(t)}$.
Let $\rho_L\equiv\rho_L(t)=\tr_{B}\ket{\psi}\bra{\psi}$,
$\sigma_L\equiv\sigma_L(t)=\tr_B\ket{\phi}\bra{\phi}$
be the reduced states on a
block $A$ of length $L$ ($B$ is the complement of $A$).
The error we make in this approximation is
$$
\epsilon\ge\left\|\ket{\psi}\bra{\psi}-\ket{\phi}\bra\phi\right\|_1
\ge \|\rho_L-\sigma_L\|_1=:2T\ ,
$$
where we have used that the trace norm $\|\eta\|_1=\tr\,|\eta|$ is
contractive under the partial trace.
  We now exploit Fannes' continuity inequality for the von Neumann entropy
  in its improved version by
Audenaert~\cite{audenaert},
\begin{align*}
|S(\rho_L)-S(\sigma_L)|\le& T\log(2^L-1)+H(T,1-T)\\
    \le& T L +1\ ,
\end{align*}
where $H(T,1-T)=-T\log T-(1-T)\log(1-T)\le1$ is the binary entropy,
and thus
\begin{align*}
S(\sigma_L)&\ge S(\rho_L)- T L -1\\
    &\ge S(\rho_L)- \tfrac12\epsilon L -1\ .
\end{align*}
Using the lower bound on $S(\rho_L(t))\equiv S_L(t)$ of
Theorem~\ref{thm:entropy}, we have (under the corresponding conditions on
$N$, $L$, and $t$) that
\[
S(\sigma_L(t))\ge \frac{4}{3\pi}t-\tfrac12\epsilon L-\tfrac12\ln t -2\ .
\]
To get an optimal bound, we choose $L$ as small as possible for
the given $t$,
$L=4t/e$, which implies the new constraint $t\ge 5e/2$,
and obain
\begin{equation}
\label{eq:approx-entropy-lbnd}
S(\sigma_L(t))\ge \left(\frac{4}{3\pi}-\frac{2\epsilon}{e}\right)t
    -\tfrac12\ln t -2\ .
\end{equation}

In order to turn this into a lower bound on the size of an MPS
$\ket{\phi(t)}$ used to approximate the state at time $t$, note that the
entropy of any block is naturally upper bounded by the rank of the reduced
state $\rho_L$, which in turn is at
most $D^2$. Here, $D$ is the so-called \emph{bond dimension} of the MPS,
which is polynomially related to the complexity of the
MPS~\cite{schollwoeck}.  Then, (\ref{eq:approx-entropy-lbnd}) gives
\begin{equation}
\label{eq:logD-lin-lowerbnd}
\log D \ge \left(\frac{2}{3\pi}-\frac{\epsilon}{e}\right)t
    -\tfrac14\ln t -1\ ,
\end{equation}
i.e.\ the bond dimension -- and thus the resources for a straight
forward MPS simulation of the time evolution -- will scale
exponentially with time as soon as the accuracy
$\epsilon<\epsilon_0=2e/3\pi\approx0.577$.

Note that the same argument holds similarly for lattices in higher
spatial dimensions. As long as the entropy $S(\rho_A(t))$ of a
region $A$ has a linear leading term, there is an accuracy
$\epsilon_0$ beyond which $S(\sigma_A(t))$ has to grow linearly as
well. Hence, approaches based on projected entangled pair states
(PEPS)\cite{PEPS} or variants thereof are burdened with the same
problem.

\section{Proof of Theorem \ref{thm:entropy}
\label{sec:proof}}

We now prove Theorem~\ref{thm:entropy}. Let us briefly sketch the
main steps: First, we map the spin chain to a one-dimensional
system of free fermions, which can be solved exactly, and bound
the error made by going to the thermodymamic limit
$N\rightarrow\infty$.  The entropy of a contiguous subblock of
length $L$ is equal to the entropy of the corresponding fermionic
modes; using a parabola as a bound on the binary entropy we can
lower bound the block entropy by a quadratic function in the
correlation matrix. Combining this with the exact solution for the
thermodynamic limit, we obtain a lower bound on the block entropy
which is essentially of the form
\[
\sum_{n=1}^{2L}n^3 \frac{J_n^2(2t)}{(2t)^2}\ ,
\]
with $J_n$ Bessel functions of the first kind.  In the appendix we
bound this sum by deriving a lower bound for $L=\infty$ and an
upper bound on the error made by extending the sum. Taking all
this together then proves Theorem~\ref{thm:entropy}.

Before starting with the proof, let us recall the conditions
imposed, namely $N\ge20$, $L\ge10$, $4\le t\le eL/4$, and $t\le
N/5$. In addition we can assume $L\le N/2$ as the overall state is
pure.

\subsection{Diagonalization of the Hamiltonian}

We start by mapping the spin operators to fermionic
Majorana operators  via the Jordan-Wigner transformation
\begin{equation}
\label{eq:jordan-wigner}
c_n=\prod_{l=1}^{n-1}\sigma_z^{(l)}\sigma_x^{(n)}\ ,\qquad
c_{n+N}=\prod_{l=1}^{n-1}\sigma_z^{(l)}\sigma_y^{(n)}\ ,\qquad
n=0,\dots,N-1\ .
\end{equation}
We will refer to $c_0,\dots,c_{N-1}$ as the position and
to $c_N,\dots,c_{2N-1}$ as the momentum operators; each pair $c_k$,
$c_{N+k}$ defines a fermionic mode.
For the initial state $\ket{\psi(0)}=\ket{1\cdots 1}$ we have
\[
\ket{\psi(0)}=\prod_{k=1}^N\tfrac12(\openone+ic_kc_{k+N})=
    \lim_{\beta\rightarrow\infty} \frac{
    e^{+i\beta \sum_kc_kc_{k+N}}}{2\cosh(\beta)}\ ,
\]
and the Hamiltonian
\begin{equation}
\label{eq:ferm-ham}
\mathcal H=\tfrac i2\left[\sum_j(c_j-c_{j+1\mathrm{\;mod\;}N})c_{j+N}\right]
\end{equation}
is equal to the spin Hamiltonian (\ref{eq:spinham}) of
Theorem~\ref{thm:entropy} on the $-1$ eigenspace of
$\Pi=\sigma_z^{(1)}\cdots \sigma_z^{(N)}$, while for the $+1$ eigenspace,
the coupling between spins $N$ and $1$ has opposite sign.
Since $\Pi\ket{\psi(0)}=-\ket{\psi(0)}$ for odd $N$ and $[\mathcal
H,\Pi]=0$, (\ref{eq:ferm-ham}) indeed describes the evolution of the
Hamiltonian of Theorem~\ref{thm:entropy}. Note that the following results
equally hold for even $N$ where they describe an Ising system with a
flipped coupling across the boundary.

Since the initial state is a Gaussian and the Hamiltonian a quadratic
function in the Majorana operators, the system remains in a Gaussian state
throughout its evolution and is thus completely characterized by its
second moments, the antisymmetric \emph{correlation
matrix}~\cite{bravyi,michael:fermionic-gauss}
$$
\Gamma_{kl}=-\tfrac{i}{2}\tr\left(\rho[c_k,c_l]\right)
$$
which for the initial state is
\[
\Gamma_0=\left(\begin{array}{cc}0&-\openone\\
    \openone&0\end{array}\right)\ .
\]
It evolves as
\begin{equation}
\label{eq:timeevol}
\Gamma_t=e^{-Ht}\Gamma_0e^{Ht}\ ,
\end{equation}
where we defined the
Hamiltonian matrix $H$ by
\[
\mathcal H=\frac{i}4\sum_{k,l=0}^{2N-1} H_{kl}[c_k,c_l]\ .
\]
We now apply Fourier transformations
$ \mathcal F_{kl}=e^{2\pi i\,kl/N}/\sqrt{N}$
to both the position and the momentum subspace.
Thereby, $\Gamma_0$ and $H$ become block diagonal
with blocks
\begin{equation}
\label{eq:intro:fourier-h-and-gamma0}
\hat\Gamma_0(\phi_k)=\left(\begin{array}{cc}0&-1\\1&0
    \end{array}\right)
\end{equation}
and
\begin{equation}
\hat H(\phi_k)=\frac12\left(\begin{array}{cc}
    0&1-e^{i\phi_k}\\-1+e^{-i\phi_k} & 0 \end{array}\right)\ ,
\end{equation}
where $\phi_k=2\pi \frac kN$.
 The evolution (\ref{eq:timeevol}) becomes
\begin{equation}
\label{eq:intro:fourier-time-evol}
\hat\Gamma_t(\phi_k)= \exp\left[-\hat H(\phi_k)t\right]
    \hat\Gamma_0(\phi_k) \exp\left[\hat H(\phi_k)t\right]\ ,
\end{equation}
such that the CM at time $t$ (written in modewise ordering) is
\begin{equation}
\label{eq:gamma-t-blockform}
\Gamma_t=\left(\begin{array}{cccc}
    \gamma_0&\gamma_1&\cdots&\gamma_{N-1}\\
    \gamma_{-1}&\gamma_0&\ddots&\vdots\\
    \vdots&\ddots&\ddots&\vdots\\
    \gamma_{1-N}&\cdots&\cdots&\gamma_0
\end{array}\right)
\end{equation}
with
\begin{equation}
\label{eq:intro:gamma-n}
\gamma_n=
\frac1N\sum_{k=0}^{N-1} e^{i n \phi_k}
\hat\Gamma_t(\phi_k)\ .
\end{equation}
Solving (\ref{eq:intro:fourier-time-evol}), we find that
$$
\gamma_n=\left(\begin{array}{cc}f_n&-g_n\\g_{-n}&-f_n\end{array}\right)
$$
with
\begin{eqnarray}
\label{eq:def-fn-finite}
f_n&=&\frac 1N\sum_{k=0}^{N-1}\tfrac{i}{2}e^{in\phi_k}
    \left[e^{-i\phi_k/2}+e^{i\phi_k/2}\right]
    \sin(2t\sin\textstyle\frac{\phi_k}{2})\ ,\\
\label{eq:def-gn-finite}
g_n&=&\frac 1N\sum_{k=0}^{N-1}\tfrac12e^{in\phi_k}\left[
    1-e^{i\phi_k}+(1+e^{i\phi_k})\cos(2t\sin\textstyle\frac{\phi_k}{2})\right]
    \mathrm{d}\phi\ .
\end{eqnarray}
This is the well-known exact solution of the Ising model as found, e.g.,
in~\cite{sachdev,cardy}.

\subsection{The thermodynamic limit}

In order to facilitate the evaluation of $f_n$ and $g_n$,
we consider the limit
$N\rightarrow\infty$, i.e.\ we replace
\[
\frac1N\sum_{k=0}^{N-1}h(\phi_k)\quad\longrightarrow\quad
\frac{1}{2\pi}\int_{0}^{2\pi}h(\phi)\mathrm{d}\phi\ .
\]
We label the functions obtained thereby by
$f_n^\infty$ and $g_n^\infty$,
respectively.
The error induced by taking the limit
can be bounded using the following theorem.

\begin{theorem}
[cf.\ \cite{kress:numerical-analysis}]
Let $h:\mathbb R\rightarrow\mathbb R$ be analytic and $2\pi$-periodic. Then
there exists a strip $D=\mathbb R\times (-a,a)\subset \mathbb C$ with
$a>0$ such that $h$ can be extended to a bounded holomorphic and
$2\pi$-periodic function $h:D\rightarrow \mathbb C$. Let
$M=\sup_{D}|h|$. Then,
$$
\epsilon:=
    \left|\frac{1}{N}\sum_{j=0}^{N-1}
    h(\textstyle\frac{2\pi j}{N}\displaystyle)-
    \frac{1}{2\pi}\int_0^{2\pi}h(\phi)\mathrm{d}\phi\right|
    \le\frac{2M}{e^{a N}-1}\ .
$$
\end{theorem}
It is straightforward to see that for any $a>0$, the addends in  both
(\ref{eq:def-fn-finite}) and (\ref{eq:def-gn-finite})
are holomorphic on $D=\mathbb R\times (-a,a)$
and bounded by
\[
M=\frac14e^{|n|a}(1+e^a)(3+\exp[t(e^{a/2}+e^{-a/2})])\ .
\]
Applying the above theorem for $a=2.9$, and taking into account that $n\le
N/2$, $t\ge4$, and $N\ge20$,
one obtains that
$|f_n-f_n^\infty|\le\epsilon$, $|g_n-g_n^\infty|\le\epsilon$, where
\begin{equation}
\label{eq:N-infty-error}
\epsilon\le30 e^{-1.45N+4.5t}\ .
\end{equation}
Taking the limit
in (\ref{eq:def-fn-finite}) and
(\ref{eq:def-gn-finite}), we find
\begin{align}
\nonumber
f_n^\infty
  &=\frac{-1}{2\pi}\!\!\int\limits_{0}^{\pi}\!\!
    \left[\sin((2n-1)\alpha)+\sin((2n+1)\alpha)\right]
    \sin(2t\sin\alpha)\,\mathrm{d}\alpha\\\nonumber
  &=-\frac{1}{2}\left[J_{2n-1}(2t)+J_{2n+1}(2t)\right]\\
\label{eq:fn-inf}
  &=-\frac{2nJ_{2n}(2t)}{2t}
\end{align}
and
\begin{align}\nonumber
g_n^\infty
  &=\frac{1}{2\pi}\int_0^{\pi}\left[\cos(2n\alpha)
    +\cos((2n+2)\alpha)\right]\cos(2t\sin\alpha)
    \,\mathrm{d}\alpha+I_n\\\nonumber
  &=\frac{1}{2}\left[J_{2n}(2t)+J_{2n+2}(2t)\right]+I_n\\
\label{eq:gn-inf}
  &=\frac{(2n+1)J_{2n+1}(2t)}{2t}+I_n\ ,
\end{align}
where $I_0=\tfrac12$, $I_{-1}=-\tfrac12$, and $I_n=0$ otherwise.
Here, $J_k$ are Bessel functions of the first kind, and
we have used the recurrence relation
$\frac12[J_{n-1}(z)+J_{n+1}(z)]=nJ_n(z)/z$.

\subsection{Lower bound for the block entropy}

The reduced state of the first $L$ sites is again Gaussian and its
correlation matrix is given by the $2L\times2L$ upper diagonal
subblock of $\Gamma_t$, which we label by $A_t$. Its entropy can
be computed from the normal mode decomposition~\cite{latorre}
\[
A_t=O\; \bigoplus_{l=1}^L \left(\begin{matrix} 0&\lambda_j\\-\lambda_j&0
    \end{matrix}\right)\; O^T\ ,
\]
where $O\in \mathrm{SO}(2L)$, as
\[
S_L(t)=\sum_{l=1}^L h(\lambda_j)\ ,
\]
with
\[
h(x)=-\frac{1+x}2\log \frac{1+x}2-\frac{1-x}2\log \frac{1-x}2\ .
\]
We lower bound $h(x)\ge 1-x^2$ as
in~\cite{fannes,michael:fermionic-gauss} and thus obtain the bound
\begin{align*}
S_L(t)&\ge \sum_{l=1}^L 1-\lambda_j^2\\
&= L+\tfrac12 \tr[A_t^2]\ .
\end{align*}
Writing
$$
\Gamma_t=\left(\begin{array}{cc}
    A&C\\-C^T&B\end{array}\right)
$$
in the $1,\dots,L$ and $L+1,\dots,N$ partitioning (thus $A\equiv A_t$
etc.) and using that the
purity of the
overall state
implies $\Gamma_t^2=-\openone$~\cite{michael:fermionic-gauss,bravyi},
we find
$A^2-CC^T=-\openone$ and thus
$$
S_L(t)\ge \textstyle\frac{1}{2}\tr[CC^T]=
\frac12\sum_{kl}C^2_{kl}\equiv\frac12\|C\|_2^2\ .
$$
We further bound the sum by only considering the lower left and the
upper right corner of $C$, i.e., all entries $\gamma_n$ in
(\ref{eq:gamma-t-blockform}) with
$1\le n\le L$ or $N-L\le n\le N-1 $ (equivalently, $-L\le n\le -1$,
since $\gamma_{N-n}=\gamma_{-n}$),
and obtain
\begin{equation}
S_L(t)\ge\frac{1}{2}\sum_{n=-L}^L |n|\|\gamma_n\|_2^2\ .
\label{eq:intro:sum-n-gamma-n}
\end{equation}
We have
\begin{align*}
\|\gamma_n\|_2^2
    &=2(f_n)^2+(g_n)^2+(g_{-n})^2\\
    &\ge\underbrace{2(f_n^\infty)^2+(g_n^\infty)^2+
    (g_{-n}^\infty)^2}_{=:A_n}-
    \underbrace{2\epsilon(2|f_n^\infty|+|g_n^\infty|+|g_{-n}^\infty|)}_{=:B_n}
\end{align*}
and with $|J_n|\le1/\sqrt{2}$ for $|n|\ge1$,  
\[
B_n\le\epsilon\left(4\sqrt2\frac{|n|}{t}+2|I_{-|n|}|\right)\
,\quad |n|\ge1\ .
\]
Thus, the total error made in (\ref{eq:intro:sum-n-gamma-n})
is bounded by
\[
\frac12\sum_{n=-L}^L |n|B_n\le\epsilon
\left(\frac{2\sqrt2}{3}\frac{L(L+1)(2L+1)}{t}+1\right)\le0.3
\]
where the latter follows using (\ref{eq:N-infty-error}) together with
$L\le N/2$, $4\le t\le N/5$, and $N\ge20$. Thus,
\begin{align}
\nonumber
S_L(t)&\ge\frac{1}{2}\sum_{n=-L}^L |n|A_n-0.3\\
\nonumber
    &\ge\sum_{k=1}^{2L} \frac{k^3 J_k^2(2t)}{(2t)^2}
    -\frac {J_1(2t)}{2t}+\frac14-0.3\\
    &\ge\sum_{k=1}^{2L} \frac{k^3 J_k^2(2t)}{(2t)^2}\,-\,0.14\ .
\label{eq:ent-lowerbnd-besselsum}
\end{align}
where we have used Eqs.~(\ref{eq:fn-inf}) and (\ref{eq:gn-inf}),
$J_{-n}(z)=(-1)^nJ_n(z)$, 
$|J_1(2t)|\le1/\sqrt2$, and $t\ge4$. We now bound the sum in
(\ref{eq:ent-lowerbnd-besselsum}) using
Lemma~\ref{lemma:bessel-sum-lowerbnd} and
\ref{lemma:besselsum-upperbnd} (see
Appendix
) and obtain
\begin{align*}
S_L(t)&\ge
    \frac{2}{3\pi}(2t)
    -\tfrac12\ln(2t)
    -\frac{4-\pi}{4\pi}
    -\frac{3\pi-4}{12\pi(2t)^2}
    -L\left(\frac{e\,t}{2L}\right)^{4L}
    -0.14\\
&\ge
    \frac{4}{3\pi}t
    -\tfrac12\ln t
    -1
\end{align*}
using $4\le t\le eL/4$ and $L\ge10$, which proves Theorem~\ref{thm:entropy}.
\hspace*{\fill}$\square$

Note that different constraints on $t$, $L$, and $N$ can be
chosen, resulting
in different corrections to the leading terms in the expansion of
$S_L(t)$.

\section{Discussion}
We gave a rigorous proof to the observation that the entanglement entropy
in a simple one-dimensional system increases essentially linearly.
Although the proof is tailored to a single exactly solvable model, we
think that this is a widespread property, thus being a requirement
which has to be met by every classical simulation method. Of course,
there always exist specific cases with tailored representations
(e.g. the quasi-free Fermions used here or Clifford circuits) 
close to which one can meet this requirement%
~\cite{freefermions,cliffordgates}    -- even MPS
allow
for a description of specific states with only polynomial effort in the
entanglement entropy (e.g., if
the matrices are tensor products\cite{Strings}). Yet, these
methods seem much more restricted in their applicability than MPS,
so that
 based on present
knowledge it seems unlikely that there is an efficient and generally
applicable classical simulation method. Note that similar results can be
found by considering any R\'enyi entropy with $\alpha>1$ instead of the
von Neumann entropy~\cite{mps-entropies}.  Also observe that in our
argument we required a good \emph{global} approximation of the state which
is a priori not necessary for obtaining good results for observables with
only local non-trivial support.

\section*{Acknowledgements}

We thank B.\ Horstmann for helpful discussions.  This work has
been supported by the EU (projects COVAQIAL and SCALA), the
cluster of excellence Munich Centre for Advanced Photonics (MAP),
the DFG-Forscher\-gruppe 635, and the Elite Network of Bavaria
project QCCC.

\section*{Appendix: Bounds on Bessel sums $\sum n^3 J_n^2(z)$
\label{sec:bessel-sum-bounds}}

\begin{lemma}
\label{lemma:bessel-sum-lowerbnd}
For $z\ge1$,
\begin{equation}
\label{eq:bessel-sum-lowerbnd}
Q(z):=\sum_{n=1}^\infty n^3J_n^2(z)
\ge \frac{2}{3\pi}z^3
-\tfrac12z^2\ln z
-\frac{4-\pi}{4\pi}z^2
-\frac{3\pi-4}{12\pi}\ .
\end{equation}
\end{lemma}
\begin{proof}%
Differentiating $Q$ using the recurrence relations
\begin{equation}
\label{eq:rec-relations}
\begin{aligned}
2J_n'(z)&=J_{n-1}(z)-J_{n+1}(z)\\ 
J_n(z)&=\frac{z}{2n}[J_{n-1}(z)+J_{n+1}(z)]
\end{aligned}
\end{equation}
yields
\begin{eqnarray}
\nonumber
Q^\prime(z)&=&
    \frac{z}{2}\sum_{n=1}^\infty n^2 [J^2_{n-1}(z)-J^2_{n+1}(z)]\\
&=&\frac{z}{2}\Bigg[J_0^2(z)+
    \underbrace{\sum_{n=1}^{\infty}4nJ_n^2(z)}_{=:f(z)}\;
    \Bigg]\ .
\label{eq:besbnd:1st-deriv}
\end{eqnarray}
To obtain a lower bound on $f(z)$, we differentiate again and find
\begin{eqnarray}
f'(z) &=&2z\sum_{n=1}^{\infty} [J^2_{n-1}(z)-J^2_{n+1}(z)]\nonumber\\
&=&2z[J_0^2(z)+J_1^2(z)]\ .\label{eq:besbnd:fprime}
\end{eqnarray}
As we prove in Lemma~\ref{lemma:j0-j1-lowerbnd} later on,
$$
J_0^2(z)+J_1^2(z)\ge\frac{2}{\pi z}-\frac{1}{z^2},\qquad z\ge1\ .
$$
Bounding $J_0^2(z)+J_1^2(z)\ge 0$ for $0\le z\le 1$, we obtain by
integration
$$
f(z)\ge \frac{4(z-1)}{\pi}-2\ln z
$$
for $z\ge1$ and $f(z)\ge0$ otherwise. Together with $J_0^2(z)\ge0$, we have
$Q'(z)\ge zf(z)/2$ and by integration, (\ref{eq:bessel-sum-lowerbnd})
follows for $z\ge1$.
\end{proof}
Note that the proof can be extended straightforwardly to sums
$\sum_{n=K}^\infty$.

\begin{lemma}
\label{lemma:besselsum-upperbnd}
For $z\le eK/4$, $K\ge2$,
\begin{equation}
\label{eq:besselsum-upperbnd}
\tilde Q(z)=\sum_{n=K+1}^\infty n^3J_n^2(z)\le
    \frac{Kz^2}{2} \left(\frac{ez}{2K}\right)^{2K}\ .
\end{equation}
\end{lemma}
\begin{proof}%
As in the proof on Lemma~\ref{lemma:bessel-sum-lowerbnd}, we differentiate
$\tilde Q$ using the recurrence relations (\ref{eq:rec-relations}) and
find
\[
\tilde Q'(z)=\frac z2\Big[
    (K+1)^2 J_{K}^2(z)+(K+2)^2 J_{K+1}^2(z)+
    \underbrace{\sum_{n=K+2}^\infty 4nJ_n^2(z)}_{=:\tilde f(z)}
\Big]
\]
and
\[
\tilde f'(z)=2z\left[J_{K+1}^2(z)+J_{K+2}^2(z)\right]\ .
\]
We now use the lower bound $|J_n(z)|\le (z/2)^n/n!$, 
which after integration leads to
\[
\tilde Q(z)\le C(K,z)\left(\frac z2\right)^{2K+2}
\]
with a prefactor $C(K,z)$ of order $(1/K!)^2$,
which for $z\le eK/4$,
$K\ge2$, and using
$k!\ge\int_k^\infty e^{-t}t^{k}\mathrm{d}t\ge(k/e)^k$,
straighforwardly reduces to (\ref{eq:besselsum-upperbnd}).
\end{proof}

\begin{lemma} For $z\ge1$,
\label{lemma:j0-j1-lowerbnd}
\[J_0^2(z)+J_1^2(z)\ge \frac{2}{\pi z}-\frac{1}{z^2}\ .
\]
\end{lemma}
Similar bounds can be proven for any pair $J_n^2+J_{n+1}^2$
along the same lines. Note that the asymptotic scaling
$J_n^2(z)+J_{n+1}^2(z)\sim2/\pi z$ is well known~\cite{watson:besselbook}.
\begin{proof}%
We use the expansion (\cite{watson:besselbook}, Sec.\ 7.3)
$$
J_n(z)=\displaystyle\sqrt{\frac{2}{\pi z}}
    \left[\cos(z-\frac{n\pi}2-\frac\pi4\displaystyle)
    \left(\sum_{m=0}^{p-1}
    \frac{(-1)^m(n,2m)}{(2z)^{2m}}+P(z)\right)\right.
$$
$$
    \qquad\qquad-\left.
    \sin(z-\frac{n\pi}2-\frac\pi4\displaystyle)
    \left(\sum_{m=0}^{q-1} \frac{(-1)^m(n,2m+1)}{(2z)^{2m+1}}+Q(z)
    \right)\right]
$$
with
$$
(n,m):=\frac{\Gamma(n+m+1/2)}{m!\Gamma(n-m+1/2)}\ ,
$$
where the remainders $P$ and $Q$ lie between zero and the first neglected
term of the respective series, given that $z>0$, $2p>n-\tfrac12$, and $2q\ge
n-\tfrac32$. With $J_n(z)=S_n(z)+E_n(z)$,
where $S_n(z)$ comes from the series and
$E_n(z)$ from the remainders, we have $|J_n(z)|\ge
|S_n(z)|-|E_n(z)|$ and thus
\begin{equation}
J_n^2(z)\ge S_n(z)^2-2|S_n(z)||E_n(z)|\ .
\label{eq:butterstollen}
\end{equation}
We choose $p=q=1$, and first upper bound
$|S_n(z)|\le\frac{2}{\sqrt{\pi z}}$  ($n=0,1$)
for $z\ge\tfrac38$ and then, using (\ref{eq:butterstollen}),
\[
J_0^2(z)+J_1^2(z)\ge\frac{2}{\pi z}-\frac{1}{\pi z^2}-
    \frac{3}{2\sqrt{2}\pi z^3}+\frac{1}{32\pi z^3}-
    \frac{45}{32\sqrt{2}\pi z^4}\ge
    \frac{2}{\pi z}-\frac{1}{z^2}
\]
for $z\ge1$, using $|\sin(z)|\le1$, $|\cos(z)|\le1$.
\end{proof}

\end{document}